\documentclass[11pt]{article}

\usepackage{fullpage}
\usepackage{amsthm}
\usepackage{times}
\usepackage{graphicx}
\usepackage{amssymb}
\usepackage{float}
\usepackage{url}
\usepackage{enumerate}
\usepackage[boxed,section]{algorithm}
\usepackage{algpseudocode}
\usepackage{color}
\usepackage[table]{xcolor}
\usepackage{subfig}
\usepackage{epsfig} % needed for new xfig version
\usepackage{graphics}
\usepackage{amsfonts}
\usepackage{amssymb}
\usepackage[cmex10]{amsmath}
\usepackage{algorithm}
\usepackage{algorithmicx}
\usepackage[fleqn,tbtags]{mathtools}
\usepackage{amsthm}
\usepackage{tikz}
\usepackage{verbatim}
\usetikzlibrary{arrows}
\usetikzlibrary{decorations.markings}
\usepackage{bbding}
\usepackage[stable]{footmisc}

\newtheorem{theorem}{Theorem}[section]
\newtheorem{lemma}[theorem]{Lemma}
\newtheorem{definition}[theorem]{Definition}

\newtheorem{remark}[theorem]{Remark}
\newtheorem*{remark*}{Remark}

\def\R{\Bbb R}
 
\def\P{\mbox{$\cal P$}} 
 
\def\Q{\mbox{$\cal Q$}} 
\def\G{\mbox{$\cal G$}}
\def\F{\Bbb F}
  
\def\ex{{{\xi}}}

\newtheorem*{rep@theorem}{\rep@title}
\newcommand{\newreptheorem}[2]{%
\newenvironment{rep#1}[1]{%
 \def\rep@title{#2 \ref{##1}}%
 \begin{rep@theorem}}%
 {\end{rep@theorem}}}

\newreptheorem{theorem}{Theorem}
\newreptheorem{lemma}{Lemma}

\newcommand{\vect}[1]{\mathbf{#1}}   %Notation for vector.
\newcommand{\vl}{\boldsymbol{\gamma}}
\newcommand{\vtau}{\boldsymbol{\tau}}
\newcommand{\vL}{\vect{L}}
\newcommand{\vF}{\vect{F}}

\newcommand{\vt}{\vect{b}}
\newcommand{\vT}{\vect{B}}

\newcommand{\vf}{\vect{f}}
\newcommand{\vU}{\boldsymbol{\Gamma'}}
\newcommand{\er}{\varphi}
\newcommand{\vx}{\vect{x}}
\newcommand{\vX}{\vect{X}}
\newcommand{\vY}{\vect{Y}}
\newcommand{\vy}{\vect{y}}
\newcommand{\vw}{\vect{w}}
\newcommand{\vz}{\vect{z}}
\newcommand{\vZ}{\vect{Z}}
\newcommand{\vone}{\vect{1}}
\newcommand{\vzero}{\vect{0}}
 
\newcommand{\ML}{{\bf ML}}
\newcommand{\LP}{{\bf LP}}

\newcommand{\E}[1]{\mathsf{E}\left[{#1}\right]}   %Expectation
\newcommand{\Ep}[2]{\mathsf{E}_{#1}\left[{#2}\right]}   %Expectation
\newcommand{\bPr}{\mathsf{Pr}}

\newtheorem*{definition*}{Definition}

\newcommand{\Normal}[2]{\mathcal{N}\left({#1}, {#2}\right)}
\begin{document}

%\tikzset{->-/.style={decoration={
%  markings,
%  mark=at position .5 with {\arrow{>}}},postaction={decorate}}}

%\tikzset{-<-/.style={decoration={
%  markings,
%  mark=at position .5 with {\arrow{<}}},postaction={decorate}}}

\title{LP decoding excess over symmetric channels}

\author{Louay~Bazzi and~Ibrahim~Abou-Faycal 
\thanks{The authors are
    with the Department of Electrical and Computer Engineering,
    American University of Beirut, Beirut 1107 2020, Lebanon (e-mails:
    $\{$lb13, ia14$\}$@aub.edu.lb).}}

\maketitle

\begin{abstract}
  We consider the problem of Linear Programming (LP) decoding of
  binary linear codes.  The LP excess lemma was introduced by 
the first author, B. Ghazi, and R. Urbanke 
(IEEE Trans. Inf. Th., 2014) as a technique to trade
  crossover probability for ``LP excess'' over the Binary Symmetric  Channel.  
  We generalize the LP excess lemma to discrete, 
  binary-input, Memoryless, Symmetric and LLR-Bounded (MSB) channels.
  As an application, we extend a result by the first author and H. Audah 
  (IEEE  Trans. Inf. Th., 2015) on the impact of redundant checks on LP
  decoding to discrete MSB channels.
\end{abstract}

\section{Introduction}

In 2003, Feldman \cite{Fel03} introduced Linear Programming (LP)
decoding as a relaxation of Maximum Likelihood (ML) decoding.  The
good performance of LP decoding of LDPC codes and its relation to
iterative decoding was established in multiple studies such
as~\cite{FWK05, FS05, FMS07, KV03} (a comprehensive survey is found
in~\cite{BA14}).

The LP excess lemma was introduced and established in~\cite{BGU14} in
the context of the Binary Symmetric  Channel (BSC) 
as a technique to trade crossover probability
for ``LP excess'' when analyzing the LP decoder error probability
under the assumption that the all zeros codeword was transmitted. The
lemma says that if the LP decoder works on a slightly
nosier channel, we can guarantee  that it corrects a slightly
shifted-down version of the received LLRs.  In dual terms, this
implies the existence of a dual witness~\cite{FMS07} where the
variable nodes inequalities are satisfied on the variable nodes with
some constant positive ``LP excess''.  The lemma was used to study the
LP decoding thresholds of spatially coupled codes~\cite{BGU14} and the impact of redundant parity checks on the LP decoding
thresholds of LDPC codes on the BSC~\cite{BA14}.

In this paper we extend the LP excess lemma from the BSC to discrete, 
binary-input, Memoryless, Symmetric and LLR-Bounded (MSB) channels.
We define the channel model in Section \ref{sc:chmod} and we give the
needed background on LP decoding in Section \ref{sc:lpdec}.  We
state and prove our main result in Section \ref{sc:lpexbc}.  As an
application, we use the extended lemma in Section \ref{sc:app} to
extend the result of~\cite{BA14} to discrete MSB channels.

\vspace{-0.15in}
\subsection{Channel model}
\label{sc:chmod}

We consider MSB channels: an {\em MSB channel}~\cite{FS05} is a {\em binary-input} {\em Memoryless} channel where the input  alphabet is $\{0,1\}$ and the transition probability
has a {\em Symmetry} property as well as a {\em Bounded LLR}
property. For simplicity of the presentation, we assume that the
channel is {\em discrete}, i.e., the output alphabet $\Sigma$ is a
finite set (or a countably infinite set).  The channel is {\em
  symmetric} in the sense that we have a partition of $\Sigma$ into
pairs $(a,a^*)$, such that $\bPr (a | 0) = \bPr(a^* | 1)$ and $\bPr (a
| 1) = \bPr(a^* | 0)$.  The {\em pairing} is a bijective map $*:\Sigma
\rightarrow \Sigma$ such that $a^{**}=a$ for each $a \in\Sigma$.  Thus
the channel is fully specified by a triplet $ch = (\Sigma, p, *)$,
where $p$ is a probability distribution on $\Sigma$ when $0$ was
transmitted, i.e., $\bPr(a|0) = p(a)$ and $\bPr(a|1) = p(a^*)$. The
Log-Likelihood-Ratios (LLR) $L_{ch}(\cdot) = L(\cdot)$ is a
real-valued map on $\Sigma$ given by 
$$L(a) =\ln{\frac{p(a)}{p(a^*)}}.$$  Note that $L(a) = - L(a^*)$ for each $a
\in \Sigma$.  We assume that the channel is {\em LLR bounded} in the
sense that $\| L \|_\infty$ is upper bounded by a constant. If
$\Sigma$ is finite, LLR boundedness is equivalent to $p(a)\neq 0$
for all $a\in \Sigma$.  We denote by $\mu_{ch} = \mu$ the 
 LLR probability distribution given $0$ is transmitted, i.e., $\mu$ is the 
probability distribution of $L(a)$ where $a$ is   sampled according to $p$.

The importance of discrete MSB channels stems from the fact that they
allow the decoder to use  soft quantized information.  They
include for example the BSC, the mixed BSC-erasure channel and the
finitely-quantized additive Gaussian-noise channel.  The binary
erasure channel is an example of a discrete symmetric channel with
possibly infinite LLRs.

We are interested in small {\em
  distortions} of discrete MSB channels: 
\begin{definition}[Channel distortion]
If $ch = (\Sigma, p, *)$ is a
discrete MSB channel and $\alpha > 0$, we call channel $ch'$ an
{\em $\alpha$-distortion} of $ch$ if $ch' = (\Sigma, p',*)$ for some
probability distribution $p'$ on $\Sigma$ such that the $L_1$-distance
$$\| p - p'\|_1 :=\sum_{a} |p(a)-p'(a)|\leq \alpha.$$ Note that, $ch'$
shares with $ch$ the same paring map $*$.
\end{definition}

For instance, consider the $\beta$-BSC channel with cross over
probability $\beta$. %$0< \beta < 1$.
An $\alpha$-distortion of the $\beta$-BSC is the $\beta'$-BSC where
$|\beta-\beta'|\leq \alpha/2$.

\begin{comment}%%123
  Another simple concrete example is the mixed BSC-erasure channel
  $ch$ with flipping probability $\beta$ and erasure probability
  $\rho$ , where $\beta, \rho>0$ are such that $\beta+\rho<1$.  An
  $\alpha$-distortion of $ch$ is another mixed BSC-erasure channel
  with parameters $\beta',\rho'$ such that $|\beta -
  \beta'|+|\rho-\rho'|+|\beta - \beta'+\rho-\rho'|\leq \alpha$.
\end{comment}

{\bf Notations.} In this document we use a bold-faced notation to
refer to $n$-dimensional vector: we transmit a length-$n$ binary
string $\vx \in \{0,1\}^n$ and receive  $\vy \in
\Sigma^n$ of $\vx$.  Additionally, we denote by $p^n$ the product
distribution on $\Sigma^n$ associated with $p$ and $\mu^n$ the product
distribution on $\R^n$ associated with $\mu$. Thus, if $\vx=\vzero$,
where $\vzero$ is the all-zeros vector, then $\vy$ is distributed
according to $p^n$ and the corresponding LLR vector $$\vl = \vL(\vy) :=
(L(\vy_i))_{i=1}^n \in \R^n$$ is distributed according to $\mu^n$.

\vspace{-0.15in}
\subsection{LP decoding} 
\label{sc:lpdec}

Let $Q \subset \F_2^n$ be an $\F_2$-linear code with blocklength $n$
and $ch = (\Sigma, p, *)$ a discrete MSB channel.  Consider
transmitting a codeword $\vx \in Q$ over $ch$, which outputs $\vy\in
\Sigma^n$.  The ML decoder of $Q$ is given by $$\ML(\vy) =
\arg\max_{\vx\in Q} {P_{\vY|\vX}(\vy|\vx)}.$$  In terms of the LLR
vector $\vl = \vL(\vy)$, the ML decoder is given by $$\ML_Q(\vl)
=\arg\min_{\vx\in Q} \langle \vx,\vl \rangle,$$ where $ \langle \vx,\vl
\rangle := \sum_ i x_ i \vl_i.$  

Feldman {\em et
  al.}~\cite{Fel03,FWK05} introduced the notion of LP decoding, which
is based on relaxing the optimization problem on $Q$ into a LP.  Due
to the linearity of the objective function $ \langle \vx,\vl \rangle
$, optimizing over $Q$ is equivalent to optimizing over the convex
polytope $\operatorname{conv}(Q) \subset \R^n$ spanned by the convex
combinations of the codewords in $Q$.  The idea of Feldman is to relax
$\operatorname{conv}(Q)$ into a larger lower-complexity polytope.  In
general terms, an {\em LP-relaxation} of $Q$ is a $Q$-symmetric convex
polytope $P \subset [0,1]^n$, where {\em $Q$-symmetry} means that
$(|\vx_i-\vy_i|)_{i=1}^n \in Q$, for each $\vx \in Q$ and $\vy \in
P$~\cite{Fel03}. Note that $Q$-symmetry implies that $Q \subset P$.

The {\em LP decoder} is
%the relaxation of the ML decoder  
given by \[
\mbox{$\LP_{P}(\vl) = \arg\min_{\vx\in P}\langle \vx,\vl \rangle.$}
\]

While useful constructions of $P$ are obtained from Tanner graph
representations~\cite{Fel03, FWK05}, it is simpler to establish the
LP-excess lemma in the general framework of $Q$-symmetric polytopes $P
\subset [0,1]^n$.  The $Q$-symmetry of $P$ implies that when
evaluating the LP decoding error probability, we can assume without
loss of generality that the all-zeros codeword $\vzero$ was
transmitted~\cite{Fel03}.  Thus $\vl\sim \mu^n$, where $\mu=\mu_{ch}$
is the LLR probability distribution given $0$.  As in previous
works~\cite{Fel03,FWK05}, we assume that the LP decoder fails if
$\vzero$ is not the unique optimal solution of the LP, i.e., the
$P$-LP decoder {\em succeeds} on $\vl$ iff $\LP_P(\vl) = \vzero$.

We say that the LP decoder {\em succeeds with LP excess $\ex$} on
$\vl$ if it succeeds on $\vl-\ex \vone$, i.e., $\LP_P(\vl-\ex \vone) =
\vzero$, where $\vone \in \R^n$ is the all ones vector and $(\vl-\ex \vone)_i = \vl_i -\ex$, for $i=1,\ldots,n$.

For constructions of $P$ from a Tanner graphs, LP excess can be
interpreted in terms of the notion of a dual-witness~\cite{FMS07} as
follows.  In dual terms, the $P$-LP decoder succeeds with LP excess
$\ex$ on $\vl$ iff $\vl-\ex \vone$ has a dual-witness, i.e., $\vl$ has
a dual-witness where each of the dual-witness variable nodes
inequalities is satisfied with ``LP excess $\ex$'' (see Definition 2.1
and Theorem 2.2 in~\cite{BA14} for the equivalent dual
characterizations of LP decoding success).

When studying the LP decoding error probability as the block length
$n$ tends to infinity, we consider an {\em infinite family} of
$\F_2$-linear codes $\Q = \{Q_n \}_n$ and an associated infinite
family of LP-relaxation $\P = \{ P_n\}_n$.  We say that the $\P$-LP
decoder succeeds on $ch$ with high probability if
\[
\mbox{$\lim_{n\rightarrow \infty}\bPr_{\vl \sim \mu^n} \left[
    \LP_{P_n}(\vl) \neq \vzero \right]=0.$}
\]
We say that the $\P$-LP decoder {\em ``succeeds on $ch$ with LP excess
  $\ex$ with high probability''\/} if $$\mbox{$\lim_{n\rightarrow
    \infty}\bPr_{\vl \sim \mu^n} \left[ \LP_{P_n}(\vl - \ex \vone) \neq
    \vzero\right]=0.$}$$

\section{LP excess lemma}\label{sc:lpexbc}

In this section, we extend the BSC LP excess lemma~\cite{BGU14} stated below to
discrete MSB channels.
\begin{lemma}[\cite{BGU14}]\label{lpexbsc} 
  {\em (BSC LP Excess Lemma: trading crossover probability with LP
    excess)} Consider the $\beta$-BSC  which crossover
  probability $0< \beta< 1/2$.  Let $\Q$ be an infinite family of
  $\F_2$-linear codes and $\P$ an associated family of LP-relaxations.\\
  Assume that there exists $\beta<\beta'< 1/2$ such that the $\P$-LP
  decoder succeeds on the $\beta'$-BSC with high probability.  \\Then,
  there exists a $\ex>0$ such that the $\P$-LP decoder succeeds on the
  $\beta$-BSC with LP excess $\ex$ with high probability.
\end{lemma}

\begin{lemma} \label{lpexg} {\em (MSB LP Excess Lemma: trading channel
    distortion with LP excess)} Let $ch$ be a discrete MSB channel,
  $\Q$ an infinite family of $\F_2$-linear codes and $\P$ an
  associated family of LP-relaxations. \\ Assume that there exists 
  $\alpha>0$ such that for each $\alpha$-distortion $ch'$ of $ch$, the
  $\P$-LP decoder succeeds on $ch'$ with high probability. \\Then, there
  exists $\ex>0$ such that the $\P$-LP decoder succeeds on $ch$ with
  LP excess $\ex$ with high probability.
\end{lemma}

In proving the Lemma, we follow similar steps to those taken
in~\cite{BGU14}. The starting point in~\cite{BGU14} is to realize the
$\beta'$-BSC as a distortion of the $\beta$-BSC resulting from the
bit-wise OR of the $\beta$-BSC error event with an independent
Bernoulli random variable $B$. The distorted channel operates
according to the original channel if $B=0$ and it produces an error if
$B=1$.  To generalize this construction, we use a similar Bernoulli-induced distortion of $ch$. The key new ingredient is a construction
of a probability distribution $q$ supported on the set of output
symbols with negative LLRs.  The distorted channel $ch'$ operates
according to the original channel if $B=0$ and according to $q$ if
$B=1$.  A key property of the constructed $q$ will be that the LLR map
$L'$ of $ch'$ is a positive constant scale of that of $ch$, i.e.,
there exists a constant $c\in (0, 1)$ such that $L'(a) = c L(a)$ for all
$a\in \Sigma$. This property will be essential in extending the
argument of~\cite{BGU14} to our setup.

\begin{proof}[Proof of Lemma \ref{lpexg}] 
  The proof is based on the fundamental cone. Let 
$C_n \subset \R^n$
  be the fundamental cone~\cite{KV03} of the $P_n$-LP decoder, i.e.,
  the set of all LLR vectors correctly decoded by the decoder: $$C_n =
  \{ \vl \in \R^n: \LP_{P_n}(\vl) = \vzero\}.$$ Since $\LP_{P_n}(\vl -
  \ex \vone) = \vzero$ is equivalent to $\vl \in C_n + \ex \vone$, our
  objective is to show that there exists a $\ex>0$ such that
  $\mu^n(C_n+ \ex \vone) = 1 - o_n(1)$. The hypothesis of the 
  theorem guarantees that for any $\alpha$-distortion $ch'$ of $ch$,
  $\mu'^n(C_n) = 1 - o_n(1)$, where $\mu'=\mu_{ch'}$ is the LLR
  probability distribution of $ch'$ given $0$.

  By the definition of the LP decoder, ${C}_n$ is the interior of the
  polar cone of ${P}_n$, i.e., $$C_n = \{ \vl \in \R^n: \langle \vl,
  \vx \rangle > 0 \mbox{ for each nonzero }\vx \in P_n\}.$$ We note that
  since ${P}_n \subset [0,1]^n \subset (\R^{+})^n$, $C_n$ is closed
  under translation by vectors in the non-negative quadrant, i.e.,
  $C_n+(\R^{+})^n \subset C_n$.  We will argue that $\ex$ exists using
  only the property that $C_n\subset \R^n$ is a convex cone such
  that $C_n+(\R^{+})^n \subset C_n$.

  Consider the partition of $\Sigma$ into three sets: 
\begin{eqnarray*}
\Sigma_- &=& \{ a
  \in \Sigma : p(a) < p(a^*)\}\\
 \Sigma_0~ &=& \{ a \in \Sigma : p(a) =  p(a^*)\}\\
 \Sigma_+ &=& \Sigma_-^*.
\end{eqnarray*}
  Thus $L$ is negative on
  $\Sigma_-$, zero on $\Sigma_0$ and positive on $\Sigma_+$. Without
  loss of generality, we assume that $\Sigma_-$ and $\Sigma_+$ are
  nonempty (otherwise, the channel capacity is zero).

  Let $0< \delta< 1$ be a constant such that $\delta \leq \alpha/2$
  and define channel $ch'=(\Sigma, p',*)$, where $p'$ is the
  distribution on $\Sigma$ given by
  \[ \begin{array}{lll}
    p'(a) & = \delta q(a) + (1-\delta)p(a) &\mbox{ if } a \in \Sigma_- \\
    p'(a) & = (1-\delta)p(a) &\mbox{ if }  a \in \Sigma_0 \cup \Sigma_+,
  \end{array} 
  \]
  % $p'(a)  = \delta q(a) + (1-\delta)p(a)$ if $a \in \Sigma_-$  and 
  % $p'(a)  = (1-\delta)p(a)$ if   $a \in \Sigma_0 \cup \Sigma_+$, 
  and where $q$ is a probability distribution on $\Sigma_-$ that will be
  specified later. We will sample from $p'$ as follows. First we
  sample a Bernoulli random variable $B \sim Ber(\delta)$ which takes the value $1$ with probability $\delta$. 
If $B=0$, we sample from $p$ and if $B=1$, we sample
  from $q$.  Channel $ch'$ is an $\alpha$-distortion of $ch$ because
  $\|p-p'\|_1\leq 2 \delta \leq \alpha$.  The LLR map of $ch'$ denoted
  by $L'$ is given by:
  \[
 L'(a) = \ln \left[ \frac{\delta q(a)+
        (1-\delta)p(a)}{(1-\delta)p(a^*)} \right]   \mbox{ if } a\in
    \Sigma_-, 
  \]
  $L'(a) = - L'(a^*)$ if $a\in \Sigma_+$ and $L'(a) = 0$ if $a\in
  \Sigma_0$. We choose $q$ so that there exists a constant  $c\in (0, 1)$ 
  such that $L'(a) = c L(a)$ for all $a \in \Sigma$ which is
  guaranteed by enforcing $L'(a) = c L(a)$ on $a \in \Sigma_-$, i.e.,
  \[
  \frac{\delta q(a)+ (1-\delta)p(a)}{(1-\delta)p(a^*)} =
    \left( \frac{p(a)}{p(a^*)} \right)^c, \qquad a\in \Sigma_-.  
  \]
  Solving for $q(\cdot)$, we get 
  \[
  q(a) = \frac{1-\delta}{\delta}p(a) \left[ \left(
        \frac{p(a^*)}{p(a)}\right)^{1-c} -1 \right], \qquad a\in
    \Sigma_-.
  \]
  Since $c\in (0,1)$ %$0 < c< 1$ 
and $p(a^*)> p(a)$ for $a\in \Sigma_-$, we have
  $q(a) > 0$ on $a \in \Sigma_-$.  To guarantee that $\sum_a q(a) =
  1$, we choose  $c\in (0, 1)$ so that $s(c) = \frac{\delta}{1-\delta}$,
  where $$s(c) := \sum_{a\in \Sigma_-} p(a) \left[ \left(      \frac{p(a^*)}{p(a)} \right)^{1-c} -1 \right].$$ This follows from
  the continuity of $s(\cdot)$ as a function of $c$ and the facts that
  $s(1) = 0$ and $$s(0) = \sum_{a\in \Sigma_-} p(a^*) - \sum_{a\in
    \Sigma_-} p(a) = p(\Sigma_+) - p(\Sigma_-)>0,$$ since $\Sigma_+$
  and $\Sigma_-$ are assumed to be nonempty. In what follows,
 fix  $\delta \in (0,1)$ 
to be any constant such that $\delta \leq
  \frac{\alpha}{2}$ such that $\frac{\delta}{1-\delta} < p(\Sigma^+) -
  p(\Sigma^-)$ to guarantee the existence of $q$ and $c$.

  In the remainder of the proof we follow the steps in~\cite{BGU14}:
  we use an averaging argument followed by Markov Inequality.  For
  clarity, we will use capital letters to refer to random quantities.  
  Define $\vf: \Sigma^n \times \Sigma^n \times \{0,1\}^n 
  \rightarrow \Sigma^n$ by 
\[
\vf(\vy,\vz;\vt)_i =\left\{ \begin{array}{ll}\vz_i & \mbox{ if } \vt_i = 1\\
    \vy_i & \mbox{ if  } \vt_i = 0. \end{array}\right.
\]
Thus, if $\vY \sim
  p^n$, $\vZ \sim q^n$ and $\vT \sim Ber(\delta)^n$, then
  $\vf(\vY,\vZ;\vT)$ is distributed according to $p'^n$, and $\vl'=
  \vL'(\vf(\vY,\vZ;\vT))$ is according to $\mu'^n$. For each $\vy\in
  \Sigma^n$, define the random vector
  \[
  \mbox{$\vU(\vy,\vZ;\vT) = \beta \,\vL'(\vf(\vy,\vZ;\vT)) = \beta c
    \, \vL(\vf(\vy,\vZ;\vT)) \in \R^n$}
  \]
  over the random choice of $\vZ \sim q^n$ and $\vT \sim
  Ber(\delta)^n$, where $\beta> 0$ is a constant to be specified
  later.  Denoting by $1_{C_n} : \R^n \rightarrow \{0,1\}$ the
  indicator function of $C_n$ (i.e.  $1_{C_n}(\vl) = 1$ iff $\vl \in
  C_n$), we define $\vw(\vy)\in \R^n$ for $\vy\in \Sigma^n$ by
  \begin{align}
    \vw(\vy) \, & = \Ep{\vZ, \vT}{\vU(\vy,\vZ;\vT) \times
      1_{{C}_n}(\vU(\vy,\vZ;\vT))}.
    \label{eq:w}
  \end{align}

  For each $\vy \in \Sigma^n$ we have $\vw(\vy) \in {C}_n$ since
  ${C}_n$ is a convex cone.
  % Moreover, ${\bf v}\in C_n$ for each ${\bf v}\geq \beta \vw(\vy)$
  % since $C_n+(\R^{+})^n \subset {C}_n$ (the vector inequality is to be
  % interpreted coordinate-wise).
  Thus, interpreting vector inequalities  coordinate-wise,
  \begin{equation}\label{psee}
    \mu^n({C}_n + \ex \vone)   \geq \bPr_{\vY \sim p^n} \left[ (\vL(\vY) - \vw(\vY) ) \geq \ex \vone \right]
  \end{equation}
  because ${\bf v}\geq \vw(\vy)$, for any ${\bf v}\in C_n$ and any
  $\vy\in \Sigma^n$ since ${C}_n + (\R^{+})^n \subset
  {C}_n$. Equation~(\ref{eq:w}) can be written as
  \begin{align*}
    \vw(\vy) \, & =\E{\vU(\vy,\vZ;\vT)} - \E{\vU(\vy,\vZ;\vT) |
      \er(\vy,\vZ;\vT)} \cdot \Phi(\vy),
  \end{align*}
  where $\er(\vy,\vZ,\vT)$ is the error event ``$\vU(\vy,\vZ;\vT) \not
  \in {C}_n$'' and $$\Phi(\vy) := \bPr_{\vZ,\vT}[\er(\vy;\vZ,\vT)].$$
  The first term $$\E{\vU(\vy,\vZ;\vT)_i} = \beta c (1-\delta) L(\vy_i)
  - \beta c \delta s,$$ where $$s := - \Ep{Z\sim q}{L(Z)}$$ is a positive
  scalar because $q$ is supported on $\Sigma_-$ and $\Ep{Z\sim
    q}{L(Z)}$ is strictly negative. The second term
  $$\E{\vU(\vy,\vZ;\vT)_i | \er(\vy,\vZ;\vT)} \geq - \beta
  \|L'\|_\infty = - \beta c \|L\|_\infty$$ since the LLRs are bounded.
  It follows that
  \[
  \vw(\vy)_i \leq \beta c \bigl[ (1-\delta) L(\vy_i) -\delta s +
  \|L\|_\infty \Phi(\vy) \bigr].
  \]
  Setting $\beta = \frac{1}{c(1-\delta)}$, we get $$\vw(\vy)_i \leq
  L(\vy_i) - \frac{\delta s- \|L\|_\infty \Phi(\vy)}{1-\delta}.$$

  Therefore, to guarantee that the vector inequality
  $\vL(\vy)-\vw(\vy) \geq \ex \vone$ holds, it is enough to require
  the scalar inequality ${\delta s- \|L\|_\infty \Phi(\vy)} \geq
  \ex(1-\delta)$. % or equivalently, $\Phi(\vy) \leq \frac{\delta s -    \ex(1-\delta)}{\|L\|_\infty}$. 
Note this reduction of the vector
  inequality to a scalar inequality critically depends on the choice
  of $q$ so that $L' = cL$.  Setting $\ex = \frac{\delta
    s}{2(1-\delta)}$, we get from (\ref{psee}) that  
  \[
  \mbox{$1-\mu^n({C}_n + \ex\vone) \leq \bPr_{\vY}\left[\Phi(\vY) >
      \frac{\delta s }{2\|L\|_\infty} \right].$}
  \]
  Using Markov Inequality, and the fact that $\Ep{\vY}{\Phi(\vY)} = 1
  - \mu'^n({C}_n)$, we obtain
  % $ $1-\mu^n({C}_n + \ex\vone) \leq \frac{2\|L\|_\infty}{\delta s}(1
  % - \mu'^n({C}_n)).$
  \[
  \mbox{ $1-\mu^n({C}_n + \ex\vone) \leq \frac{2\|L\|_\infty}{\delta
      s}(1 - \mu'^n({C}_n)).$ }
  \]
  % \begin{eqnarray*}
  %   1-\mu^n({C}_n + \ex\vec{1}) &\leq& \bPr_{\vY}\left[\Phi(\vY) > \frac{\delta s }{2\|L\|_\infty} \right]\\
  %   &\leq& \frac{2\|L\|_\infty}{\delta s}(1 - \mu'^n({C}_n)).
  % \end{eqnarray*}
  Since $\mu'^n({C}_n)=1-o_n(1)$, we conclude that $\mu^n({C}_n +
  \ex\vone) = 1 - o_n(1)$, where $\ex>0$ is constant which depends on
  $\alpha$ and the channel $ch$.
\end{proof}

\begin{remark}\label{unbllr} 
\begin{itemize}
\item[I)] If we replace probability distributions with densities,
  the LP excess lemma and its proof hold for continuous MSB
  channels. 
\item[II)] We conjecture that the LLR boundedness is not
  needed for the lemma to hold.
  \begin{comment}%%123
    To get rid of this requirement in the above proof, one has to find
    a lower bound on $\E{\vU(\vy,\vZ;\vT)_i | \er(\vy,\vZ;\vT)}$ other
    than $- c \|L\|_\infty$.
  \end{comment}
  One justification of this conjecture is the Gaussian channel
  discussed below.
\end{itemize}
\end{remark}

\vspace{-0.15in}
\subsection{Gaussian channel}

On the $\sigma$-Additive White Gaussian Noise ($\sigma$-AWGN) channel,
we receive $Y=(-1)^x+\sigma Z$, where $x=$ $0$ or $1$ is the
transmitted bit and $Z \sim \Normal{0}{1}$, the standard Gaussian
distribution.  The AWGN has unbounded LLRs.  

By a simple scaling
argument, the following version of the LP excess lemma holds on the
AWGN:

\begin{lemma}
  Let $Q_n \subset \F_2^n$ be an $\F_2$-linear code, $P_n\subset \R^n$
  an LP-relaxation of $Q_n$ and $\sigma'>\sigma>0$.  The probability
  of success of the $P_n$-LP decoder on the $\sigma'$-AWGN is equal to
  its probability of success on the $\sigma$-AWGN with LP excess
  $\ex$, where $\ex = \frac{\sigma'-\sigma}{\sigma'} $.
\end{lemma}

\begin{proof}
  The LLR  map is $L(y) = \frac{2}{\sigma^2}y$ (e.g., \cite{Fel03}).
  Assume that $\vzero$ was transmitted and let $\mu$ and $\mu'$ be the
  LLR densities associated with $\sigma$ and $\sigma'$, respectively.
  Since $$\frac{\sigma}{\sigma'}(1+\sigma' z) =1+ \sigma z - \ex ,$$ we
  get $\mu'^n(C_n) = \mu^n(C_n+ \ex \vone)$, for each $C_n \subset
  \R^n$ closed under multiplication by positive scalars and in
  particular for the fundamental cone $C_n$ of the $P_n$-LP decoder.
\end{proof}

The  distinguishing features of  the AWGN from other channels in this context are:
(1) scaling $Z$ corresponds to distorting the channel and (2) the LLR
map is linear in $y$.

\section{Application to redundant parity checks}\label{sc:app}

The BSC LP excess lemma was used in~\cite{BA14} to
show that the LP decoding threshold of LDPC codes on the BSC remains
the same upon adding all redundant parity checks, assuming that the
underlying Tanner graph has {\em bounded degree} and possesses two
natural properties called {\em asymptotic strength} and {\em rigidity}
(see Corollary 1.7 in \cite{BA14}). One implication of this result is
that the BSC threshold is a function of the dual code and is not tied
to the particular Tanner graph realization of the
  code. We use in this section our extension of the LP excess lemma to
  extend the result of~\cite{BA14} from the BSC to discrete MSB
  channels:
\begin{theorem}\label{appth}
  Let $\G=\{G_n\}_n$ be an infinite family of Tanner graphs, where
  $G_n$ has $n$ variable nodes.  Let $\overline{\G} = \{
  \overline{G_n}\}_n$ be the resulting family of Tanner graphs
  obtained by adding all redundant checks, i.e., the parity check
  nodes of $\overline{G_n}$ correspond to all the nonzero elements of
  the dual code of $G_n$.  Assume that $\G$ has bounded check degree
  and that $\G$ is asymptotically strong and rigid.
  Let $ch$ be a discrete MSB channel.  Assume that there exists 
  $\alpha>0$ such that for each $\alpha$-distortion $ch'$ of $ch$, the
  $\overline{\G}$-LP decoder succeeds on $ch'$ with high
  probability. Then, the $\G$-LP decoder succeeds on $ch$ with high
  probability.
\end{theorem}

In order to prove the theorem we only need the following extension of
Theorem 1.2 in \cite{BA14} to discrete MSB channels:

\begin{lemma}\label{th1.2}
  Let $\G, \overline{\G}, ch, \alpha$, $ch'$ be as in Theorem
  \ref{appth}, and Let $d$ be the maximum degree of a check node in
  $\G$.   For $k \geq d$, let $\overline{\G}^k := \{
    \overline{G}_n^{k} \}_n$ be the resulting family of Tanner
  graphs obtained by including all redundant checks of degree at most
  $k$.  There exists a sufficiently large constant $k\geq d$ --where
  $k$ depends on $\alpha$ and the channel only-- such that the
  $\overline{\G}^k$-LP decoder succeeds on $ch$ with high probability.
\end{lemma}

\begin{proof}[Proof of Theorem \ref{appth}]
  Following the proof of Corollary 1.7 in \cite{BA14}, Theorem
  \ref{appth} follows from Lemma \ref{th1.2} and the rigidity of $\G$
  which implies that for each constant $k\geq d$, the LP decoding
  polytope $P\bigl( \overline{G}_n^k \bigr) = P(G_n)$ for $n$ large enough.
\end{proof}

\begin{proof}[Proof of Lemma \ref{th1.2}]
  We use below the terminology of the proof Theorem 1.2 in \cite{BA14}
  to explain the needed modifications.  At a high level, the following
  changes are needed: 
\begin{itemize}
\item
 Instead of a variable received correctly or
  in error, we have positive or nonpositive LLRs respectively.
\item 
  The value of LP excess is $\ex$ instead of $\frac{\delta}{4}$.
\item
  The maximum absolute value of a received LLR is the {\em constant}
  $\|L\|_\infty$ instead of $1$.
\end{itemize}
  More specifically, consider operating the $\overline{G_n}$-LP
  decoder on $ch$: assume that the all-zeros codeword was transmitted
  and consider the received LLR vector $\vl\sim \mu^n_{ch}$.  By the
  LP excess lemma, there exists a constant $\ex>0$ (dependent on
  $\alpha$) such that with high probability, the $\overline{G_n}$-LP
  decoder corrects $\vl$ with LP excess $\ex$, i.e., it corrects $\vl
  - \ex \vone$.  In what follows, consider any such $\vl \in \R^n$.
  To verify Lemma \ref{th1.2}, we will show that the
  $\overline{G}_n^k$-LP decoder corrects $\vl$ for a sufficiently
  large constant $k\geq d$ which depends on $\ex$ and the channel (and
  does not depend on $n$).  For notational simplicity, we will denote
  $G_n,\overline{G}_n^k$ and $\overline{G_n}$ by $G,\overline{G}^k$
  and $\overline{G}$, respectively. Also, let $E, \overline{E}^k$ and
  $\overline{E}$ be the set of edges of $G,\overline{G}^k$ and
  $\overline{G}$, respectively.

  By Theorem 2.2 in \cite{BA14}, there is a hyperflow
  $w:\overline{E}\rightarrow \R$ in $\overline{G}$ for $\vl - \ex
  \vone$.  Hence, $$\vF(w)< \vl - \ex\vone,$$ where $\vF(w) \in \R^n$ is the
  flow as specified in Definition 2.1 in \cite{BA14}.  Let $$V^+ = \{ i
  : \vl_i - \ex > 0\}$$ and $$V^{-}= \{ i : \vl_i - \ex \leq 0\}$$ be the
  set of variables nodes with positive and nonpositive ``shifted LLR''
  respectively.  Since $\overline{G}$ contains all redundant checks,
  we can assume by Lemma 4.2 in \cite{BA14} that $w$ is primitive,
  hence the inflow to each variable in $V^{+}$ is zero and the outflow
  from each variable in $V^{-}$ is zero.  Following \cite{BA14},
  define the {\em trimmed hyperflow} and the resulting {\em risky} and
  {\em problematic} variables as follows.  Trim $w$ by removing all
  check nodes of degree larger than $k$. The trimming process leads to
  a distorted dual witness $w^k: \overline{E}^k\rightarrow \R$ in
  $\overline{G}^k$.  The {\em problematic} variables nodes are those
  for which the hyperflow variables nodes inequalities of $w^k$ are
  violated with respect to $\vl$.  A variable node is called {\em
    risky} if it receives at least $\frac{\ex}{2}$ flow from the
  removed check nodes, thus all the problematic variables are risky.
  The set of risky variable nodes is called $U$. We have $U \subset
  V^{-}$ since $w$ is primitive.  Hence 
\[
\begin{array}{llll}
\vF_i(w^k)&\leq& 0 & \mbox{ if } i\in U, \mbox{ and }\\  
\vF_i(w^k)&<& \vl_i - \ex/2 & \mbox{ if } i\not\in U.
\end{array}
\]
  Since all the removed checks have degree larger than $k$ and since
  $\vl_i \leq \|L\|_\infty$ for each $i$, the removed checks give the
  variables in $V^{-}$ at most $$\frac{| V^{+}|. ( \|L\|_\infty - \ex)
  }{k-1}\leq \frac{n\|L\|_\infty}{k-1}$$ flow.  It follows that
  $$|U|\leq \frac{2n\|L\|_{\infty}}{\ex(k-1)}.$$  Since $w$ is
  primitive, to fix $w^k$ on the problematic variables, it is enough
  to give each variable in $U$ an $\|L\|_\infty$ flow. Following
  \cite{BA14}, we do that by exploiting the asymptotic strength of
  $\G$ and the remaining excess on the nonrisky variable nodes. The
  remaining LP excess on each nonrisky variable is at least
  $\ex-\frac{\ex}{2}=\frac{\ex}{2}$.  Consider the asymmetric LLR
  vector $\vtau\in \R^n$ given by: 
   \[                        
   \vtau_i = \left\{\begin{array}{ll} -\| L\|_\infty &\mbox{ if }
       i\in U \\ \frac{\ex}{2} & \mbox{
       otherwise. } \end{array}\right.
   \]
  We use the remaining excess to fix $w^k$ by superposing $w^k$ with a
  dual witness for $\vtau$.

  Since $\G$ is {\em asymptotically strong}, there exists a constant
  $\delta>0$ dependent on $\frac{\ex}{2\|L\|_\infty}$ such that if
  $|U|\leq \delta n$, the LP decoder of $G$ succeeds on
  $\frac{\vtau}{\|L\|_\infty}$ and hence on $\vtau$. Thus, if
  $$\frac{2\|L\|_\infty}{\ex (k-1)} \leq \delta,$$ then $\vtau$ has a
  dual witness $v: E \rightarrow \R$ in $G$. Since $k \geq d$, let
  $v^k: \overline{E}^k\rightarrow \R$ be the extension of $v$ to
  $\overline{G}^k$ by zeros.  Thus $\vF(v^k)< \vtau$ and accordingly
  $$\vF(w^k+v^k)< \vl.$$ Therefore, $w^k + v^k$ is the desired dual
  witness of $\vl$ in $\overline{G}^k$.  It follows (from Theorem 2.2
  in~\cite{BA14}) that the $\overline{G}^k$-LP decoder successfully
  corrects $\vl$.

  In summary, there exists a constant $\delta>0$ dependent on
  $\frac{\ex}{2\|L\|_\infty}$ such that if $$k = \max \left\{ d, \left\lceil
    \frac{2\|L\|_\infty}{\ex\delta} \right\rceil +1 \right\},$$ which
  depends on the $\ex$ and the channel, then the $\overline{G}^k$-LP
  decoder corrects $\vl$ for any $\vl\in \R^n$ such that the
  $\overline{G}$-LP decoder corrects $\vl - \ex \vone$.
\end{proof}

Note that the proof of Lemma \ref{th1.2} breaks down if the LLRs are
unbounded even if Lemma \ref{lpexg} holds for channels with unbounded
LLRs (see Remark \ref{unbllr}.II).


\begin{thebibliography}{1}



\bibitem%[Fel03]
{Fel03}
J. Feldman. {\em Decoding error-correcting codes via linear programming.}
PhD thesis, Massachusetts Institute of Technology,
2003.


\bibitem%[FWK05]
{FWK05} 
J. Feldman, M.J. Wainwright, and D.R. Karger. Using linear
programming to decode binary linear codes. {\em Information Theory,
IEEE Transactions on}, vol. 51, no. 3, pp. 954-972, 2005.



\bibitem%[FS05]
{FS05} 
J. Feldman and C. Stein. LP decoding achieves capacity. In
{\em Proceedings of the Sixteenth Annual ACM-SIAM Symposium on
Discrete Algorithms},  pp. 460-469, Philadelphia,
PA, USA, 2005. 



\bibitem%[FMS+07]
{FMS07} 
J. Feldman, T. Malkin, R.A. Servedio, C. Stein, and M.J. Wainwright. LP decoding corrects a constant fraction of errors.
{\em Information Theory, IEEE Transactions on}, vol. 53, no. 1, pp. 82-89, 2007.




\bibitem%[KV03]
{KV03} 
R. Koetter and P.O. Vontobel.  Graph-covers and iterative
decoding of finite length codes. In {\em Proceedings of the IEEE
International Symposium on Turbo Codes and Applications},
pp. 75-82, 2003.




\bibitem%[BGU14]
{BGU14} 
L. Bazzi, B. Ghazi, and R.L. Urbanke. 
Linear Programming Decoding of Spatially Coupled Codes. 
{\em Information Theory, IEEE Transactions on}, vol. 60, no. 8, pp. 4677-4698, 2014.



\bibitem%[BA14]
{BA14} L. Bazzi and H. Audah. 
Impact of redundant checks on the  LP decoding thresholds of LDPC codes. 
{\em Information Theory, IEEE Transactions on}, vol. 61, no. 5, pp. 2240-2255, 2015.

\end{thebibliography}
\end{document}